\documentclass[11pt,conference]{IEEEtran}
\usepackage{times}
\usepackage{epsfig}
\usepackage{amsmath}
\usepackage{amsfonts}
\usepackage{graphicx}
\usepackage{amssymb}
\usepackage{amstext}
\usepackage{latexsym}
\usepackage{color}
\usepackage{ifthen}
\usepackage{multirow}
\usepackage{subfigure}
\newcommand{\be}{\begin{equation}}
\newcommand{\ee}{\end{equation}}
\newcommand{\bear}{\begin{eqnarray}}
\newcommand{\eear}{\end{eqnarray}}
\newcommand{\bears}{\begin{eqnarray*}}
\newcommand{\eears}{\end{eqnarray*}}
\newcommand{\bi}{\begin{itemize}}
\newcommand{\ei}{\end{itemize}}
\newcommand{\ben}{\begin{enumerate}}
\newcommand{\een}{\end{enumerate}}

\newcommand{\beq}{\begin{equation}}
\newcommand{\eeq}{\end{equation}}

\newcommand{\lp}{ \left(}
\newcommand{\rp}{ \right)}

\newtheorem{theorem}{Theorem}[section]

\newtheorem{lemma}[theorem]{Lemma}

\newcommand{\xbf}{\mbox{${\bf X }$} }
\newcommand{\ybf}{\mbox{${\bf Y }$} }

\newcommand{\Prob}{\mbox{${\mathbb P}$} }

\newcommand{\prob}[1]{\Prob \left\{ #1 \right\}}

\newcommand{\SNR}{{\sf SNR}}

\begin{document}

\title{Diversity-Multiplexing Tradeoff of the Half-Duplex Relay Channel}

\author{\authorblockN{Sameer Pawar}
\authorblockA{Wireless Foundations\\
 UC Berkeley,\\
 Berkeley, California, USA.\\
{\sffamily sameerpawar@berkeley.edu}} \and
\authorblockN{Amir Salman Avestimehr}
\authorblockA{Wireless Foundations\\
 UC Berkeley,\\
 Berkeley, California, USA.\\
{\sffamily avestime@eecs.berkeley.edu}} \and
\authorblockN{David N C. Tse}
\authorblockA{ Wireless Foundations\\
 UC Berkeley,\\
 Berkeley, California, USA.\\
{\sffamily dtse@eecs.berkeley.edu}}
 }

% make the title area
\maketitle

\begin{abstract}
We show that the diversity-multiplexing tradeoff of a half-duplex single-relay channel
with identically distributed Rayleigh fading channel gains meets the $2$ by $1$ MISO bound.
We generalize the result to the case when there are $N$ non-interfering relays and show that the diversity-multiplexing tradeoff is equal to the $N+1$ by $1$ MISO bound.
\end{abstract}

\section{Introduction}
Cooperation between nodes can provide both diversity and degree of
freedom gain in wireless fading channels \cite{SenErkAaz1,
SenErkAaz2, LTW1}. The diversity-multiplexing tradeoff (DMT) was a
metric introduced by Zheng and Tse \cite{ZheTse} to evaluate
simultaneously the diversity and degrees of freedom gain in general
fading channels. Significant effort has been spent in the past few
years in computing the DMT of cooperative relay networks. The
simplest such network has one relay and a direct link between the
source and the destination (Fig. \ref{fig:relay_channel}, with the
channel gains modeled as quasi-static identically distributed
Rayleigh faded and known only to the respective receive node. A
simple upper bound to performance is the DMT of the $2$ by $1$ MISO
channel obtained when the source and relay can fully cooperate to
transmit to the destination:
\[
d(r) = 2(1-r) \ \ \ \ \ \ \ \ 0\leq r \leq 1
\]

It is quite easy to see that this upper bound can be achieved if the relay can operate on a full-duplex mode, i.e. transmit and receive at the same time. But most radios can only operate on a half-duplex mode. Somewhat surprisingly, the DMT for the half-duplex single-relay network is still an open problem despite substantial effort.

\begin{figure}[ht]
\begin{center}
\includegraphics[width = 1\columnwidth]{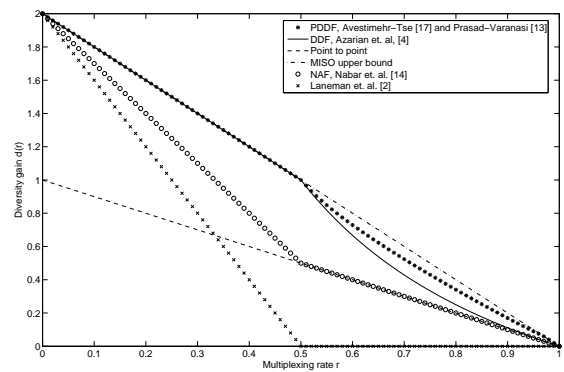}
\caption{Diversity multiplexing tradeoff of several schemes for the half-duplex relay channel}\label{fig:state_of_art}
\end{center}
\end{figure}

Figure \ref{fig:state_of_art} shows the DMT performance of several schemes and how they compare to
the MISO bound. We see that none of the schemes achieves the bound for the entire range of
multiplexing gains. The dynamic-decode-and-forward \cite{AzGamSch} and partial decode-and-forward
\cite{PraVar,AT} schemes achieves the MISO DMT for multiplexing gains $r \le 0.5$ but there is a gap
for $r > 0.5$. Is this gap fundamental or is there a better scheme?

In this paper, we show that indeed there is a scheme that achieves the MISO DMT for {\em all} multiplexing gains $r$ up to $1$.  The problem with decode-and-forward schemes is that for $r > 0.5$, it takes too long for the relay to decode the whole message and there is not enough time for it to forward information. The problem with partial-decode-and-forward scheme is that the source does not know how to split the overall message without knowing the instantaneous channel gains of the various channels. In contrast, the scheme
that we propose, which we call {\em quantize-and-map}, does not decode or partially decode the message. Instead, the relay extracts the significant bits of the received signal above noise level by quantization and re-encodes them to forward to the relay. The destination then combines the received signal from the relay and the direct signal from the source to solve for the information bits. Because there is no need to decode any message, there is also no need for any dynamic adaptation of the listening period for the relay. In fact, it turns out that it suffices for the relay to always listen half of the time and talk half of the time regardless of the channel state.

The quantize-and-map scheme is based on a recent deterministic approach to approximate the capacity of Gaussian relay networks \cite{ADT071,ADT072,ADTISIT08,ADTIT08}. Inspired by the optimal scheme that was found for the deterministic relay networks \cite{ADT071,ADT072}, the quantize-and-map scheme was shown in \cite{ADTISIT08,ADTIT08} to achieve within a constant gap of the capacity of arbitrary Gaussian relay networks, where the constant gap does not depend on the channel parameters. A key observation is that since the scheme does not require any channel information at the nodes, it can also be utilized in a fading scenario in which there is no channel state information available at the transmitter. Now since at high SNR and high rates the approximation gap is negligible, as a corollary one can show that for any listen-transmit schedule, this scheme achieves the diversity-multiplexing tradeoff of the cut-set bound on the capacity. The desired result is obtained when this fact is combined with the observation that the DMT of the cutset bound of the half duplex network matches that of the MISO bound when the relay listens half of the time and talks the other half. This result can also be generalized to more than $1$ relay when these relays have no link between themselves.

% Our contribution
%\subsection{Results} Main contributions of this paper are

%\begin{theorem}\label{thm:result1} Diversity multiplexing tradeoff of a half-duplex
%single relay channel in section~\ref{sec:relay_channel_model} is
%\[
%d(r) = 2(1-r) \ \ \ \ \ \ \ \ 0\leq r \leq 1
%\]
%\end{theorem}
%In the interest of having a uniform solution the above theorem is
%also generalized for $N$ number of isolated (non-communicating among
%themselves) relays.

%\begin{theorem} For a general Gaussian relay network with half-duplex constraint diversity multiplexing tradeoff
%of cut-set upper bound on the capacity
%$\overline{C}_{hd}(\overline{t})$ for any fixed scheduling
%$\overline{t}$ can be achieved i.e.,
%\[
%d(r,\overline{t}) = d_{\overline{C}_{hd}(\overline{t})}(r)
%\]
%\end{theorem}
%$\overline{t}$ determines the fixed-scheduling, exact meaning of
%fixed scheduling will be made clear in subsection~\ref{subsec:Enc}.

\section{System Model}\label{sec:relay_channel_model} Consider a
network as shown in Figure \ref{fig:relay_channel} with a source
$S$, a destination $D$, and one relay node $R$.

\begin{figure}[h]
\begin{center}
\input{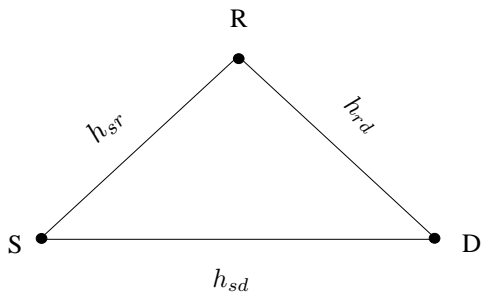}
\caption{The relay channel. \label{fig:relay_channel}}
\end{center}
\end{figure}

All the channel links $h_{sd},h_{sr},h_{rd}$ are assumed to be
flat-fading, i.i.d complex normal $\mathcal{CN}(0,1)$ distribution. It is assumed
that although random, once realized, channel gains remain unchanged
for the duration of the codeword and change independently from one
codeword to another i.e., quasi-static fading. Noise at all of the
receivers is additive i.i.d $\mathcal{CN}(0,1)$ independent of any
other form of randomness in the system. All nodes have single
antenna and have equal average power constraint specified by average
Signal to Noise Ratio (SNR), denoted by $\rho$. Relay node $R$ is assumed
to be in half-duplex operation and for simplicity it is assumed that
transmission of source and relay are synchronous at symbol level. Furthermore, channel state information (CSI) is only available at the receivers. So, relay has CSI about $h_{sr}$, destination
has CSI about $h_{sd},h_{rd}$ and no CSI at all at the source.

\section{Diversity-multiplexing tradeoff of the half-duplex relay channel}
\label{sec:single_relay} In this section we characterize the
diversity-multiplexing tradeoff of the half-duplex relay channel, described in section \ref{sec:relay_channel_model}.
First we describe the quantize-map relaying scheme that we proposed
earlier in \cite{ADTISIT08} and \cite{ADTIT08}. As we showed in
these references, this relaying scheme achieves a rate within a
constant gap to the cut-set upper bound of the capacity of the relay
channel for all channel gains, where the constant is independent of the channel
$\SNR$s. Furthermore, since this relaying scheme does not require any
channel information at the source and the relay, it can also
be performed in our scenario (i.e. no CSI at the transmitter). Now, since
at high $\SNR$ and high rates the approximation gap is negligible,
as a corollary we will show that this scheme achieves the
diversity-multiplexing tradeoff of the cut-set bound on the capacity
for any listen-transmit scheduling at the relay. Finally we
illustrate that a fixed scheduling that relay listens only half the
time and transmits the rest is enough to achieve the
diversity-multiplexing tradeoff of the $2 \times 1$ MISO channel, hence we find the optimal DMT of the half-duplex relay channel.

\subsection{Description of the relaying scheme}
\label{subsec:Enc} We have a single source $S$ with a sequence of
messages $w_k \in\{1,2,\ldots,2^{KTR}\}$, $k=1,2,\ldots$ to be
transmitted. At both the relay and the source we create random
Gaussian codebooks. Source randomly maps each message to one of its
Gaussian codewords and sends it in $KT$ transmission times (symbols)
giving an overall transmission rate of R. Due to half-duplex nature
of the relay, it has to do listen-transmit cycles. Relay operates over
blocks of time $T$ symbols and since total length of codeword at
source is $KT$ we have $K$ blocks in each codeword. Relay listens to
the first $Tt$ ($0\leq t \leq 1$) time symbols of each block. Let
$X^{1(k)}_{S}$ denote the sequence of these $Tt$ symbols transmitted at the source in block $k$. Also let $\ybf_R^{(k)}$ and $\ybf^{1(k)}_D$ be the received signal at relay
and destination respectively during this time. Then the relay it quantizes its received signal in the first $tT$ time symbols to $\hat{\ybf}_R^{(k)}$ which is then randomly mapped into a
Gaussian codeword $\xbf_R^{(k)}$ using a random mapping function
$f_R(\hat{ \ybf}_R^{(k)})$ and sends it in the next $T(1-t)$ time
symbols. Let $\ybf^{2(k)}_D$ denote the sequence of symbols received
by destination during this time. Given the knowledge of all the encoding functions at the relay and
signals received over $K$ blocks, the decoder D, attempts to decode the message sent by the source.

\subsection{DMT of the relaying scheme}\label{sec:DMT_relay_channel}
For any fixed listen-transmit scheduling strategy (i.e. fixed $t$), the cut-set upper bound on the capacity of the half-duplex Gaussian relay channel, $\overline{C}_{hd}$, is given by  (\ref{eq:CutSet}) on the top of next page \cite{khojastepour} .

\begin{table*}
\begin{eqnarray} \label{eq:CutSet}  \overline{C}_{hd}(h_{sr},h_{sd},h_{rd},\rho,t) &=& \max_{p(x_S^{1},x_S^{2},x_R)} \min  \{ t I(X_S^1;Y_R,Y_D^1|X_R)+ (1-t)I(X_S^2;Y_D^2|X_R),  t I(X_S^1; Y_D^1)+(1-t) I(X_S^2,X_R; Y_D^2)  \} \\ \nonumber &\leq& \min \{ t\left(\log (1 + \rho(|h_{sr}|^2 + |h_{sd}|^2))\right)  + (1-t)\left(\log(1 + \rho |h_{sd}|^2\right) , (1-t)\left(\log (1 + \rho(|h_{rd}|  + |h_{sd}|)^2)\right)  \\ \label{eq:CutSetE} && \quad \quad  + t\left(\log(1 + \rho |h_{sd}|^2\right) \}
\end{eqnarray}
\end{table*}

Now, as we showed in \cite{ADTIT08}, for any fixed listen-transmit
scheduling, the quantize-map relaying scheme described in Section \ref{subsec:Enc}, uniformly
achieves a rate within a constant gap to the capacity. Therefore by
Theorem 4.7 in \cite{ADTIT08}, for all channel gains we have,
{\small
\beq \label{eq:QMRate}
\overline{C}_{hd}(h_{sr},h_{sd},h_{rd},\rho,t) -\kappa \leq
R_{\text{quantize-map}}(h_{sr},h_{sd},h_{rd},\rho,t)  \eeq} where
$\kappa \leq 15$ is a constant that does not depend on the channel
gains and $\SNR$.

Now since this relaying scheme does not require any channel information at the source and the relay, it can also be performed in our scenario in which there is no channel state information available at the transmitter. Furthermore, as at high $\SNR$ and high data rates the approximation gap is negligible, as a corollary we will now show that for any fixed listen-transmit scheduling, this scheme achieves the diversity-multiplexing tradeoff of the cut-set bound.

\begin{theorem}\label{thm:achievability} For any fixed scheduling $t$, the quantize-map relaying scheme achieves the diversity-multiplexing tradeoff of $\overline{C}_{hd}$, where $\overline{C}_{hd}$ is defined by (\ref{eq:CutSet}).
\end{theorem}
\begin{proof}
Assume a targeted communication rate $R$. By (\ref{eq:QMRate}), we know  that the destination will be able to decode the information sent by the source as long as
\beq \overline{C}_{hd}(h_{sr},h_{sd},h_{rd},\rho,t) -\kappa>R \eeq
Therefore for any scheduling $t$, we have
\beq \label{eq:outageUB} \mathcal{P}_{\text{outage}}(\rho) \leq \prob{\overline{C}_{hd} -\kappa<R}\eeq where the probability is calculated over the randomness of channel gain realizations.
Now by definition, for any scheduling $t$,  the achievable diversity of quantize-map scheme is
\begin{eqnarray}
d_{QM}(r) &=& - \lim_{\rho \rightarrow \infty} \frac{\log \lp \mathcal{P}_{\text{outage}}(\rho) \rp }{\log \rho} \\
& \stackrel{(\ref{eq:outageUB})}{\geq} & - \lim_{\rho \rightarrow \infty} \frac{\log \lp  \prob{\overline{C}_{hd} -\kappa<r \log \rho} \rp }{\log \rho}  \\
& \stackrel{*}{\approx} & - \lim_{\rho \rightarrow \infty} \frac{\log \lp  \prob{\overline{C}_{hd} <r \log \rho} \rp }{\log \rho}
\end{eqnarray}
where $*$ is true since $\kappa$ is a constant and does not scale with $\rho$. Therefore for any $t$, the quantize-map relaying strategy achieves  the diversity-multiplexing tradeoff of $\overline{C}_{hd}(t)$.
\end{proof}

Next, we will show that with $t=0.5$, the diversity-multiplexing
tradeoff of $\overline{C}_{hd}$ matches the diversity-multiplexing
tradeoff of the $2\times1$ MISO channel, and hence we complete the
proof of our main Theorem.

First we give some intuition on why this is true.
\begin{figure}[h]
\begin{center}
\input{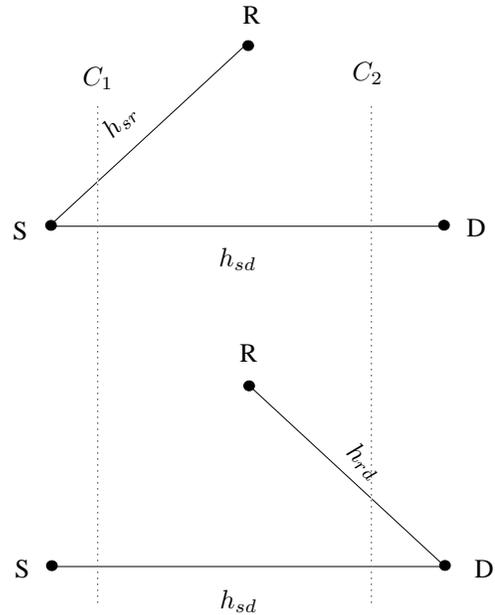}
\caption{Two scheduling modes of the system: relay listens $t$ fraction of the time and relay transmits $(1-t)$fraction of the time}
\label{fig:cut_set}
\end{center}
\end{figure}

First note that in equation (\ref{eq:CutSetE}), the first term corresponds to the
information flowing through cut $\{S\},\{R,D\}$ (see
Figure \ref{fig:cut_set}) and the second term corresponds to the
information flowing through the cut $\{S,R\},\{D\}$. Now, the value of the first
cut $\{S\},\{R,D\}$ corresponds to the capacity of a SIMO system with $1$ transmit
antenna and $2$ receive antennas where one receive antenna
(corresponding to relay) is listening only $t$ amount of time. Similarly,
the value of the second cut i.e., $\{S,R\},\{D\}$ corresponds to the capacity of a MISO system
with $2$ transmit antennas and $1$ receive antenna, where one
transmit antenna (corresponding to relay) is transmitting only $1-t$
amount of time. Since we are limited by the minimum of these two
values optimal strategy is to try to make them equal. Also since DMT
of $1 \times 2$ SIMO is same as that of $2 \times 1$ MISO, a natural choice is to set $t=0.5$.

Once we set $t=0.5$, DMT of cut-set bound is just DMT of a $2 \times
1$ MISO system with $1$ transmit antenna being used only half the
time, but this system is {\it{strictly better}} than a system with
$2$ transmit $1$ receive antennas and where each of the two transmit
antennas are used only half the time in an alternate fashion i.e., parallel
channel with rate $r$ on each channel. It is well known and easy to
compute that DMT of this parallel channel is $2(1-r)$. Also we have
obvious upper bound of DMT of MISO system which is again $2(1-r)$.
Thus the cut-set bound achieves the optimal DMT for $t=0.5$. The formal proof of this is given in Appendix \ref{app:AchMISO}.

\section{Extension to multiple-relay network}
In this section we extend our result to general multiple-relay networks. The listen-transmit scheduling model that we use to study this problem is the same as \cite{khojastepour}. In this model the network has finite modes of operation. Each mode of operation (or state of the network), denoted by $m \in \{1,2,\ldots,M \}$, is defined as a valid partitioning of the nodes of the network into two sets of "sender" nodes and "receiver" nodes such that there is no active link  that arrives at a sender node\footnote{Active link is defined as a link which is departing from the set of sender nodes}. For each node $i$, the transmit and the receive signal at mode $m$ are respectively shown by $x_i^m$ and $y_i^m$. Also $t_m$ defines the portion of the time that network will operate in state $m$, as the network use goes to infinity.
As shown in \cite{khojastepour}, the cut-set upper bound on the capacity of the Gaussian relay network with half-duplex constraint, $C_{hd}$, is given by (\ref{eq:CutSetGen}) on the top of next page.
\begin{table*}
\beq \label{eq:CutSetGen} C_{hd} \leq \overline{C}_{hd}= \max_{\substack{p(\{x_j^m\}_{j\in\mathcal{V}, m \in \{1,\ldots,M\}})\\ t_m:~ 0 \leq t_m \leq 1, ~\sum_{m=1}^M t_m=1 }} \min_{\Omega\in\Lambda_D} \sum_{m=1}^M t_m I(Y_{\Omega^c}^m;X_{\Omega}^m|X_{\Omega^c}^m)  \eeq
\end{table*}

Now we describe the quantize-map relaying scheme that we proposed in \cite{ADTISIT08,ADTIT08} for general half-duplex relay networks.

\subsection{Description of the relaying scheme}
\label{subsec:EncGen}
We have a single source $S$ with a sequence of
messages $w_k \in\{1,2,\ldots,2^{KTR}\}$, $k=1,2,\ldots$ to be
transmitted. At all nodes we create a random
Gaussian codebook. Source randomly maps each message to one of its
Gaussian codewords and sends it in $KT$ transmission times (symbols)
giving an overall transmission rate of R. Relays operate over blocks
of length $T$ symbols. Starting from the beginning of the block each relay $i$ spends a total of $t_mT$ symbols in state $m$, $m=1,\ldots,M$. In each state, if it is assigned to listen, it receives a
sequence $\ybf_i^{(k,m)}$. Otherwise, if it is assigned to transmit, it quantizes all received signals in the previous block (i.e. $\ybf_i^{(k-1,m)}$, $m=1,\ldots,M$) to $\hat{\ybf}_i^{(k,m)}$ which is
then randomly mapped into a Gaussian codeword $\xbf_i^{(k,m)}$ using a
random mapping function $f_i(\hat{ \ybf}_R^{(k,m)})$ and sends it in
that $t_mT$ time symbols. Given the knowledge of all the encoding functions at the relay and signals received over $K+|V|-2$ blocks, the decoder D, attempts to decode the message W sent by the source.

\subsection{DMT of the relaying scheme}\label{sec:DMT_N_relay}

As we showed in \cite{ADTIT08}, for any fixed listen-transmit
scheduling, the quantize-map relaying scheme described above,
achieves within a constant gap to the capacity. Therefore by
Theorem 4.7 in \cite{ADTIT08}, for all channel gains we have, \beq \label{eq:QMRateGen}
\overline{C}_{hd} -\kappa \leq R_{\text{quantize-map}}  \eeq where
$\kappa \leq 5|\mathcal{V}|$ is a constant that does not depend on
the channel gains and $\SNR$.

Therefore similar to Theorem \ref{thm:achievability} we can show the following theorem:
\begin{theorem}\label{thm:achievabilityGen} For any fixed scheduling, the quantize-map relaying scheme achieves the diversity-multiplexing tradeoff of $\overline{C}_{hd}$, where $\overline{C}_{hd}$ is defined by (\ref{eq:CutSetGen}).
\end{theorem}

To find the optimal performance of this scheme, one should optimize over all possible scheduling strategies. In general we don't know the optimizing strategy, however as we show in Section \ref{subsec:Two_hop_Network}, in a special case of two hop network with $N$ non interfering relays, a fixed uniform scheduling (i.e. $t_m=2^{-N}$, $m=1,\ldots,2^N$) achieves the optimal DMT.

\subsection{Optimal DMT of two hop network with $N$ non-interfering half duplex relays} \label{subsec:Two_hop_Network} Consider a two hop network with single
source $S$, destination $D$ and $N$ half-duplex relays $R_i, \ \ 1
\leq i \leq N$, as shown in Figure \ref{fig:N_relay_channel}.

All the assumptions in section~\ref{sec:relay_channel_model} are
carried over with additional assumption that there is no link among
any two relays. Let $h_{sr_i}$ be the link from source $S$ to relay $i$ and $h_{r_id}$ be the link from $i^{\text{th}}$ relay to destination $D$.

Here is our main result for this relay network.
\begin{theorem} \label{thm:achievability_N_relay}  The optimal diversity-multiplexing tradeoff of a two-hop relay network, with $N$ non-interfering half-duplex relays is equal to
the diversity-multiplexing tradeoff of the $(N+1) \times 1$ MISO channel. Furthermore, it is achieved by the quantize-map relaying strategy define in Section \ref{subsec:EncGen} with fixed and uniform scheduling, $t_m=2^{-N}$, $m=1,\ldots,2^N$.
\end{theorem}
\begin{proof}
See Appendix~\ref{app:Optimal_DMT_for_N_relay}
\end{proof}

\begin{figure}[h]
\begin{center}
\input{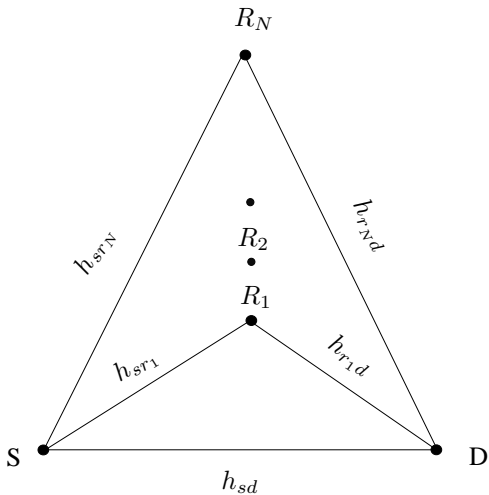}
\caption{Two-hop network with N half-duplex relays}
\label{fig:N_relay_channel}
\end{center}
\end{figure}

\appendices

\section{Achieving the MISO bound in the relay channel with $t = 0.5$}\label{app:AchMISO} From theorem~\ref{thm:achievability},
it is sufficient to show that DMT of $\overline{C}_{hd}$ is equal to
that of MISO.

For ease of computation define $n_{sd}:=\log(1 + |h_{sd}|^2\rho)$
and $\alpha_{sd}$ as its exponential order i.e.,

\[
\alpha_{sd} := \lim_{\rho\rightarrow \infty} \frac{\log(1 +
|h_{sd}|^2\rho)}{\log\rho}
\]

then the probability density function (pdf) of $\alpha_{sd}$ can be
shown to be

\begin{eqnarray*}\label{eq:alpha_pdf}
f_{\alpha_{sd}}(\alpha_{sd}) &=& \lim_{\rho\rightarrow \infty}
\exp(-\rho^{-(1-\alpha_{sd})}) \rho^{-(1-\alpha_{sd})} \log\rho\\
&=& \rho^{-(1-\alpha_{sd})}   \ \ \ \ \ \ 0 \leq \alpha_{sd} \leq 1
\end{eqnarray*}

Consider first term in inequality \eqref{eq:CutSetE}, using $t=0.5$

\begin{eqnarray*}
&&0.5\left(\log (1 + \rho(|h_{sr}|^2 + |h_{sd}|^2))\right) \\
&&+ \ 0.5\left(\log(1 + \rho |h_{sd}|^2\right) \\
&\doteq& 0.5\left(\max\{\log (1+\rho |h_{rd}|^2 ), \log (1+\rho
|h_{sd}|^2)\}\right) \\
&&+ \ 0.5\left(\log(1+\rho |h_{sd}|^2\right)\\
&=& n_{sd} + 0.5(n_{rd} - n_{sd})^+
\end{eqnarray*}
similarly second term can be simplified resulting in

\begin{eqnarray*}
\overline{C}_{hd} \doteq n_{sd} + 0.5\min\{(n_{sr} -
n_{sd})^+,(n_{rd} - n_{sd})^+\}
\end{eqnarray*}

For $R=r\log\rho$ the cut-set bound is in outage if
\begin{eqnarray*}
{n_{sd} + 0.5\min\{(n_{sr} - n_{sd})^+, (n_{rd} - n_{sd})^+\}
\leq r\log\rho}\\
i.e.,\alpha_{sd} + 0.5\min\{(\alpha_{sr} - \alpha_{sd})^+,
(\alpha_{rd} - \alpha_{sd})^+\}
\leq r\\
\end{eqnarray*}

\begin{eqnarray*}\label{eq:outage}
\therefore \mathcal{O}(r)  &=& \{ \alpha_{sd},\alpha_{sr},\alpha_{rd} \mid \alpha_{sd} \\
&&+ \ 0.5\min\{(\alpha_{sr} - \alpha_{sd})^+, (\alpha_{rd} -
\alpha_{sd})^+\} \leq r\}
\end{eqnarray*}

\begin{eqnarray*}
P_{\mathcal{O}}(r) &=& \int_{\overline{\alpha}\in
\mathcal{O}(r)}f_{\overline{\alpha}}(\overline{\alpha})d\overline{\alpha} \\
&=& \int_{\overline{\alpha}\in
\mathcal{O}(r)}f_{\alpha_{sd}}(\alpha_{sd})f_{\alpha_{sr}}(\alpha_{sr})f_{\alpha_{rd}}(\alpha_{rd})d\overline{\alpha}\\
&=& \int_{
\begin{array}{c}
  \overline{\alpha}\in
\mathcal{O}(r) \\
  0 \leq \overline{\alpha} \leq 1
\end{array}
} \rho^{-3 + (\alpha_{sd} + \alpha_{sr} + \alpha_{rd})} d\overline{\alpha} \\
&\doteq& \rho^{-d(r)}
\end{eqnarray*}
where

\begin{eqnarray*}
d(r) &=& \inf_{
\begin{array}{c}
(\alpha_{sd},\alpha_{sr},\alpha_{rd}) \in \mathcal{O}(r) \\
0\leq \alpha_{sd},\alpha_{sr},\alpha_{rd}\leq 1\\
\end{array}
} 3 - (\alpha_{sd} + \alpha_{sr} + \alpha_{rd})
\end{eqnarray*}

\begin{enumerate}
\item If $\alpha_{sd} \geq
\min\{\alpha_{sr},\alpha_{rd}\}$: Then outage implies
$\min\{\alpha_{sr},\alpha_{rd}\} \leq \alpha_{sd} \leq r$. And since
$\max\{\alpha_{sr},\alpha_{rd}\} \leq 1$ we have $(\alpha_{sd} +
\alpha_{sr} + \alpha_{rd}) \leq 1 + 2r$.

\item If $\alpha_{sd} \leq
\min\{\alpha_{sr},\alpha_{rd}\}$: Then Outage implies

\begin{eqnarray*}
\alpha_{sd} + 0.5(\min\{\alpha_{sr},\alpha_{rd}\} - \alpha_{sd})
&\leq& r\\
\alpha_{sd} + \min\{\alpha_{sr},\alpha_{rd}\} &\leq& 2r\\
\alpha_{sd} + \min\{\alpha_{sr},\alpha_{rd}\} +
\max\{\alpha_{sr},\alpha_{rd}\} &\leq& 1 + 2r
\end{eqnarray*}
\end{enumerate}
Therefore
\begin{eqnarray*}\label{eq:MISO}
d(r) &=& 3 - (1 + 2r)\nonumber\\
&=& 2(1-r)
\end{eqnarray*}
$\hfill{\blacksquare}$\\
 Thus quantize-map relaying scheme achieves the optimal DMT of
$2 \times 1$ MISO system.

\section{DMT for two-hop network with $N$ non-interfering half-duplex relays}\label{app:Optimal_DMT_for_N_relay}
We prove theorem~\ref{thm:achievability_N_relay} in two steps, we
first show that DMT of cut-set bound with fixed uniform scheduling
achieves DMT of $(N+1) \times 1$ MISO system and then apply
theorem~\ref{thm:achievability}.

\subsection{DMT of cut-set bound for fixed uniform scheduling}
In subsection~\ref{sec:DMT_N_relay} theorem~\ref{thm:achievability}
showed that for any fixed scheduling quantize-map relaying scheme
achieves the DMT of the cut-set for that scheduling. We make use of
this fact to show the achievablity of MISO performance in $N$ relay
case. First we note that since there are $N$ half-duplex relays,
each relay has a choice to be either in receiving mode or in
transmitting mode, accordingly we have $M=2^N$ states. Then
following the lead from our single relay case and using the fact the
everything in network is nicely symmetrical we operate network in
each of these states for equal amount of time i.e., $t_m = 2^{-N}$
for $m=1,\cdots,2^N$. Now if we show that DMT of cut-set for this
scheduling is equal to $(N+1)(1-r)$ we are done.

We first derive a lower bound on the cut-set. Any cut in a network
partitions all nodes into two groups $\Omega$ with $S\in \Omega$ and
its compliment $\Omega^c$ with $D\in \Omega^c$, each relay has a
choice of being in a either $\Omega$ or $\Omega_c$, thus we have
$2^N$ total possible cuts and the cut-set bound of network is equal
to the minimum of mutual information flowing through each of these
$2^N$ possible cuts.

\begin{figure}[h]
\begin{center}
\input{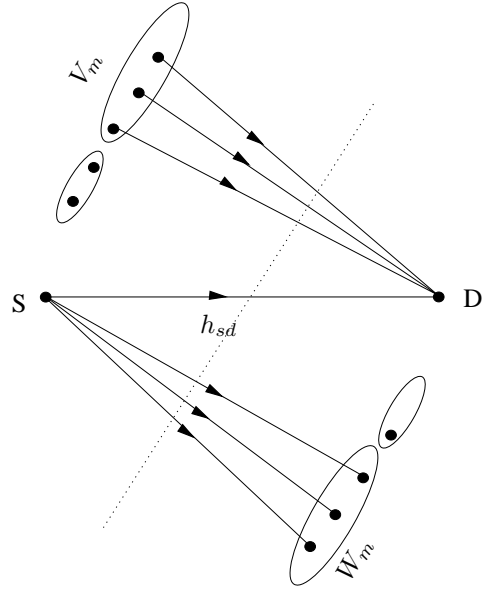}
\begin{center}
\caption{$m^{\text{th}}$ state of Network} \label{fig:N_relay_cut}
\end{center}
\end{center}
\end{figure}

\begin{figure}[h]
\begin{center}
\input{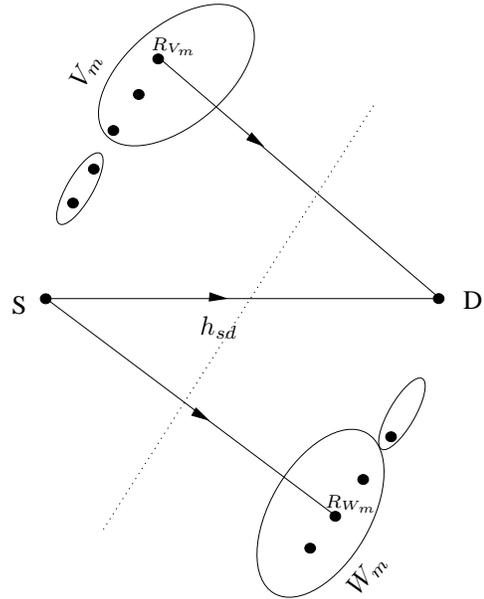}
\begin{center}
\caption{$Z-$ channel} \label{fig:Z_channel}
\end{center}
\end{center}
\end{figure}

Consider a cut $\Omega$ in the network which is operating in state
$m$ it looks as shown in fig. \ref{fig:N_relay_cut}. Let $V_m
\subseteq \Omega - S$ be the set of relays $R_j \in \Omega$ which
are transmitting and $W_m \subseteq \Omega^c - D$ be the set of
relays $R_j \in \Omega^c$ which are receiving in state $m$. Let
$R_{V_m}$ be a relay with strongest channel say $h^{*}_{rd} :=
\max_j\{h_{r_jd}\} \ j \in V_m$ to the destination and analogously
let $R_{W_m}$ be a relay with strongest channel say
$h^{*}_{sr}:=\max_j\{h_{sr_j}\} \ j \in W_m$ from source. We can
lower bound the total mutual information flowing across this cut in
fig~\ref{fig:N_relay_cut} by the the mutual information flowing
across the same cut $\{S,R_{V_m}\} \{R_{W_m},D\}$ in the Z-channel
formed by these nodes, see fig~\ref{fig:Z_channel}. This Z-channel
can be viewed as MIMO system with upper triangular channel matrix
$H= \left(
      \begin{array}{cc}
        h^{*}_{rd} & h_{sd} \\
        0 & h^{*}_{sr} \\
      \end{array}
    \right)$
So mutual information flow across this cut in Z-channel is given by
\begin{eqnarray*}
&&\log\det\left(I_{2\times2} + \rho HH^{\dag}\right) \\
&=& \log\left( 1 + \rho (|h^{*}_{rd}|^2 + |h_{sd}|^2 +
|h^{*}_{sr}|^2 + \rho|h^{*}_{rd}|^2 |h^{*}_{sr}|^2)\right)\\
&\geq& \max\{\log(1 + \rho |h_{sd}|^2), \\
&& \log\left((1+\rho|h^{*}_{sr}|^2)(1 + \rho |h^{*}_{rd}|^2)\right)\}\\
&=& \max\{\log(1 + \rho |h_{sd}|^2), \log(1+\rho \max_{j\in
V_m}(|h_{r_jd}|^2))\\
&& +
\log(1 + \rho \max_{j\in W_m}(|h_{sr_j}|^2)) \}\\
&=& \max\{n_{sd}, \max_{j\in V_m}(n_{r_jd}) +
\max_{j\in W_m}(n_{sr_j})) \}\\
\end{eqnarray*}

Thus for each cut $\Omega$ the cut value is, \beq
\label{eq:z_cutvalue} \overline{C}_{\Omega} \geq \frac{1}{2^N}
\sum_{i=1}^{2^N} \max\{n_{sd}, \max_{j\in V_m}(n_{r_jd}) +
\max_{j\in W_m}(n_{sr_j})\}
\end{equation}

Now the cut-set bound is simply, \beq
\label{eq:cut_set_N_relay}\overline{C}_{hd}\geq \min_{\Omega}
\overline{C}_{\Omega} \eeq

\begin{lemma}\label{lem:cutValueGreaterAvg}
For any cut $\Omega$, there are $N+1$ distinct links flowing across
the cut and the mutual information flowing through it given by
\eqref{eq:z_cutvalue} can be further lower bounded by their average
\beq \overline{C}_{\Omega} \geq \frac{n_{sd}+\sum_{j \in
\Omega-\{S\}} n_{r_jd}+\sum_{j \in \Omega^c-\{ D\}} n_{sr_j} }{N+1}
\eeq
\begin{proof}
See Appendix~\ref{sec:proof_lemma}
\end{proof}
\end{lemma}

\subsection{Optimality of DMT of each cut $\overline{C}_{\Omega}$}
 Following Appendix~\ref{app:AchMISO} for each $j$, we define
$\alpha_{sd},\alpha_{r_jd},\alpha_{sr_j}$ as exponential order's of
$n_{sd},n_{r_jd},n_{sr_j}$ respectively.

From lemma~\ref{lem:cutValueGreaterAvg} outage is equal to set

\begin{equation}\label{eq:outage_N_relay}
\mathcal{O}(r) = \{ \overline{\alpha}\mid \alpha_{sd}+\sum_{j \in
\Omega-\{S\}} \alpha_{r_jd}+\sum_{j \in \Omega^c-\{ D\}}
\alpha_{sr_j} \leq (N+1)r \}
\end{equation}

\begin{eqnarray*}
P_{\mathcal{O}}(r) &=& \int_{\overline{\alpha}\in
\mathcal{O}(r)}f_{\overline{\alpha}}(\overline{\alpha})d\overline{\alpha}\\
&=& \int_{
\begin{array}{c}
  \overline{\alpha}\in
\mathcal{O}(r) \\
  0 \leq \overline{\alpha}\leq 1
\end{array}
} \rho^{-(N+1)}\\
&& . \rho^{\alpha_{sd}+\sum_{j \in \Omega-\{S\}}\alpha_{r_jd} +
\sum_{j \in
\Omega^c-\{ D\}} \alpha_{sr_j}}d\overline{\alpha} \\
&\doteq& \rho^{-d(r)}
\end{eqnarray*}
where

\begin{eqnarray*}
d(r) &=& \inf_{
\begin{array}{c}
  \overline{\alpha}\in
\mathcal{O}(r) \\
  0 \leq \overline{\alpha}\leq 1
\end{array}
} (N+1)-\left(\alpha_{sd}+\sum_{j \in \Omega-\{S\}}
\alpha_{r_jd} \right.\\
&&  \left.\hspace*{1.8in} + \sum_{j \in \Omega^c-\{ D\}} \alpha_{sr_j}\right)\\
&=& (N+1)(1-r)
\end{eqnarray*}
last equality follows from the equation~\ref{eq:outage_N_relay}. Now
since each cut has optimal DMT, from inequality
\eqref{eq:cut_set_N_relay} it is clear that cut-set bound also
achieves optimal DMT. And then we use
theorem~\ref{thm:achievability}.

\subsection{Proof of
Lemma~\ref{lem:cutValueGreaterAvg}}\label{sec:proof_lemma} To prove
this, first we show the following lemma,
\begin{lemma}
\label{lem:ineq1} Consider a set of numbers $a, s_1,\ldots,s_{n}$.
Assume function $f$ is such that for any set $V \subseteq
\{1,\ldots,n \}$ we have, \beq f(V) \geq \max(a,s_V) \eeq where \beq
s_V=\{ s_i| i \in V\}\eeq Then \beq   \frac{1}{2^n} \sum_{V
\subseteq \{1,\ldots,n \}} f(V) \geq \frac{a+\sum_{i=1}^ns_i}{n+1}
\eeq
\end{lemma}
\begin{proof}
Without loss of generality assume that $s_i$'s are ordered (i.e.
$s_1 \leq s_2 \leq \ldots \leq s_n$). Then we have
\begin{eqnarray*}
&&\frac{1}{2^n} \sum_{V \subseteq \{1,\ldots,n \}} f(V) \\
& \geq & \frac{a+\max(a,s_1) + \ldots+2^{n-1}\max(a,s_n)}{2^n} \\
 &\stackrel{*}{ \geq }& \frac{a+\max(a,s_1)+\ldots+ \max(a,s_n)}{n+1} \\
& \geq & \frac{a+s_1+s_2+\ldots+s_n}{n+1}
\end{eqnarray*}
where $*$ is true by applying two sequences $(a,s_1,\ldots,s_m)$ and
$(2^{-n},2^{-n},2^{-n+1},\ldots,2^{-1})$ to the Tchebychef's
inequality,

\textbf{Tchebychef's inequality:} Assume two sequences
$(a_1,\ldots,a_n)$ and $(b_1,\ldots,b_n)$ are similarly ordered
(i.e. $(a_u-a_v)(b_u-b_v)\geq 0$, for all $u$ and $v$). Then \beq
\frac{1}{n} \sum_{i=1}^n a_ib_i \geq \lp \frac{1}{n} \sum_{i=1}^n
a_i \rp \lp\frac{1}{n} \sum_{i=1}^n b_i \rp \eeq

\end{proof}

Now we prove Lemma \ref{lem:cutValueGreaterAvg}.
\begin{proof} (proof of Lemma \ref{lem:cutValueGreaterAvg})\\
First note that for any $V_m \subseteq \Omega-\{S\}$ and $W_m
\subseteq \Omega^c-\{D\}$ we have

\begin{eqnarray*}
f(V_m,W_m)&=& \max \lp n_{sd}, \max_{i \in V_m} \lp n_{r_id}
\rp+\max_{i\in W_m}\lp n_{sr_i}\rp \rp\\
&\geq& \max(n_{sd},n_{3V_m},n_{2W_m})
\end{eqnarray*}
where
\begin{eqnarray*}
n_{3V_m}&=& \{n_{r_id}|i \in V_m \}\\
n_{2W_m}&=& \{n_{sr_j}|j \in W_m \}\\
\end{eqnarray*}
Now by Lemma \ref{lem:ineq1} we know that
\begin{eqnarray*}
&&\frac{1}{2^N} \sum_{V_m \subseteq \Omega-\{S\}} \sum_{W_m
\subseteq
\Omega^c-\{D\}} f(V_m,W_m) \\
&\geq&  \frac{n_{sd}+\sum_{i \in \Omega-\{S\}} n_{r_id}+\sum_{i \in
\Omega^c-\{ D\}} n_{sr_i}  }{N+1}
\end{eqnarray*} hence the proof is complete.
\end{proof}

\bibliographystyle{IEEEbib}

\end{document}